\documentclass[
preprint, 1p
]{elsarticle}

\usepackage[utf8x]{inputenc}
\usepackage{xcolor}
\usepackage{listings}
\usepackage{amssymb}
\usepackage{amsmath}
\usepackage{stmaryrd}
\usepackage{listings}
\usepackage{array}
\usepackage{textcomp}
\usepackage{pifont}
\usepackage{caption}
\usepackage{bbding}
\usepackage{comment}
\usepackage[
plainpages=false,
pdftitle={Automated Amortised Resource Analysis for Term Rewrite Systems},
pdfauthor={Georg Moser and Manuel Schneckenreither},
pdfsubject={Automated Amortised Resource Analysis for Term Rewrite Systems},
pdfkeywords={
  Analysis of algorithms, amortised complexity, term rewriting,
  lower bounds, automation},
]{hyperref}
\usepackage{url}
\usepackage{proof}
\usepackage{turnstile}
\usepackage{journal}

\newcommand\GM[2][r]{\ifx t#1 \textcolor{green!50!black}{[\textbf{GM:} #2]}
  \else \begin{center}\textcolor{green!50!black}{\textbf{GM:} #2} \end{center} \fi}
\newcommand\MS[2][r]{\ifx t#1 \textcolor{blue}{[\textbf{MS:} #2]}
  \else \begin{center}\textcolor{blue}{\textbf{MS:} #2} \end{center} \fi}

\begin{document}

% Environments
% \spnewtheorem*{lemma*}{Lemma}{\bfseries}{\itshape}
% \spnewtheorem*{theorem*}{Theorem}{\bfseries}{\itshape}
\newtheorem{theorem}{Theorem}
\newtheorem{corollary}[theorem]{Corollary}
\newtheorem{lemma}[theorem]{Lemma}
\newproof{proof}{Proof}
\newdefinition{definition}{Definition}
\newdefinition{example}{Example}

\newcolumntype{x}[1]{>{\centering\arraybackslash\hspace{0pt}}p{#1}}

\title{Automated Amortised Resource Analysis for Term Rewrite Systems%
% \thanks{This research is partly supported by DARPA/AFRL contract number FA8750-17-C-088.}
}

\author[1]{Georg Moser} % \corref{cor1}}
\ead{georg.moser@uibk.ac.at}
\author[2]{Manuel Schneckenreither} % \corref{cor1}}
\ead{manuel.schneckenreither@uibk.ac.at}

% \cortext[cor1]{Corresponding author}

\address[1]{
  Department of Computer Science,
  University of Innsbruck, Austria,
}
\address[2]{
  Department of Information Systems,
  Production and Logistics Management,
  University of Innsbruck, Austria,
}

\date{\today}

\begin{abstract}
% Based on earlier work on amortised resource analysis,
% we establish a novel automated amortised resource analysis for
% term rewrite systems. The method is presented in an inference system
% akin to a type system and gives rise to polynomial bounds on the innermost
% runtime complexity of the analysed term rewrite system.
% %
% Our analysis does not restrict the input rewrite system
% in any way. This facilitates integration in a general framework
% for resource analysis of programs. In particular, we have implemented
% the method and integrated it into our tool \tct.
In this paper we establish an automated amortised resource analysis for
term rewrite systems. The method is presented in an annotated type system
and gives rise to polynomial bounds on the innermost
runtime complexity of the analysed term rewrite system.
Our analysis does not restrict the input rewrite system in any way so that rewrite systems may serve
as abstractions of first-order, eagerly evaluated functional programs over user-defined inductive
data-types. This facilitates integration in a general framework for resource analysis of programs.
In particular, we have implemented the method and integrated it into our analysis tool \tct.
Furthermore, we have coupled the established analysis with a complexity reflecting transformation
from pure \OCaml\ programs. This extends the provided analysis to a fully automated resource
analysis of higher-order functional programs.
\end{abstract}


\maketitle

\section{Introduction}

Amortised resource analysis~\cite{ST:1983,Tarjan:1985} is a powerful method to assess the
overall complexity of a sequence of operations precisely. It has been established
by Sleator and Tarjan in the context of self-balancing data structures, which sometimes
require costly operations that however balance out in the long run.

For automated resource analysis, amortised cost analysis has been in particular
pioneered by Hoffmann et al., whose \raml\ prototype has grown into a highly
sophisticated analysis tool for higher-order functional programs, cf.~\cite{HDW:2017}.
In a similar spirit, resource analysis tools for imperative programs like \costa~\cite{AAGP:2011}, \cofloco~\cite{Flores-Montoya:2016} and
\loopus~\cite{SZV:2016} have integrated amortised reasoning.
In this paper, we establish a novel automated amortised resource analysis for
term rewrite systems (TRSs for short).

Consider the rewrite system $\RSa$ in Figure~\ref{fig:1}
encoding a variant of an example by Okasaki~\cite[Section~5.2]{Okasaki:1999}
(see also\cite[Example~1]{HM:2014});
$\RSa$ encodes an efficient implementation of a queue in
functional programming. A queue is represented as a pair
of two lists $\queue(f,r)$, encoding the initial part~$f$
and the reversal of the remainder~$r$. The invariant of the algorithm is
that the first list never becomes empty, which is achieved by
reversing $r$ if necessary. Should the invariant ever be violated,
an exception ($\errorHead$ or $\errorTail$) is raised.
To exemplify the physicist's method of amortised analysis~\cite{Tarjan:1985} we
assign to every queue $\queue(f,r)$ the length of $r$ as \emph{potential}. Then the
amortised cost for each operation is constant, as the costly reversal operation is
only executed if the potential can pay for the operation.
Thus, based on an amortised analysis, we may deduce the optimal linear runtime
complexity for $\RS$.

\begin{figure}[t]
\centering
\begin{alignat*}{4}
1\colon && \checkF(\queue(\nil,r)) &\to \queue(\rev(r),\nil)
&
\hspace{2ex}
7 \colon && \enq(\zero) &\to \queue(\nil,\nil)
\\
2\colon &\!& \checkF(\queue(x \cons xs,r)) &\to \queue(x \cons xs,r)
&
\hspace{2ex}
8\colon && \revp(\nil,ys) & \to ys
\\
3\colon && \tail(\queue(x \cons f,r)) &\to \checkF(\queue(f,r))
&
\hspace{2ex}
9 \colon && \rev(xs) &\to  \revp(xs,\nil)
\\
4\colon && \snoc(\queue(f,r),x) &\to \checkF(\queue(f,x \cons r))
&
\hspace{2ex}
10\colon &\!& \head(\queue(x \cons f,r)) &\to x
\\
5\colon && \revp(x \cons xs,ys) &\to \revp(xs,x \cons ys)
&
\hspace{2ex}
11\colon && \head(\queue(\nil,r)) &\to \errorHead
\\
6 \colon && \enq(\mS(n)) &\to \snoc(\enq(n),n)
&
\hspace{2ex}
12\colon && \tail(\queue(\nil,r)) &\to \errorTail
\end{alignat*}
\caption{Queues in Rewriting.}
\label{fig:1}
\end{figure}

Taking inspirations from~\cite{HM:2014,HM:2015}, the amortised analysis is based on
the potential method, as exemplified above. It employs the standard (small-step) semantics of innermost
rewriting and exploits a \emph{footprint} relation in order to facilitate the
extension to TRSs. For the latter, we suit a corresponding notion of Avanzini et
al.~\cite{AD:2017} to our context.
Due to the small-step semantics we immediately obtain an analysis which does
not presuppose termination. The incorporation of the footprint relations allows the
immediate adaption of the proposed method to general rule-based languages. The most
significant extension, however, is the extension to standard TRSs.
TRSs form a universal model of computation that underlies much of declarative
programming. In the context of functional programming, TRSs form a natural abstraction
of strictly typed programming languages like~\raml, but natively form foundations
of non-strict languages and non-typed languages as well.

Our interest in an amortised analysis for TRSs is motivated by the use of TRSs as abstract program
representation within our uniform resource analyse tool~\tct~\cite{avanzini2016tct}.
Rather than studying amortised cost analysis for a particular programming language, as provided
for example by Hoffmann et al.~in~\cite{HDW:2017}, we follow the general approach of static program analysis,
where peculiarities of specific programming languages are suitable abstracted to give way to more general constructions
like \emph{recurrence relations}, \emph{cost relations}, \emph{transition systems},
\emph{term rewrite systems}, etc.

Thus we first seek a more general
discussion of amortised analysis in (first-order) TRSs.
In this spirit we
aim at an amortised resource for TRSs in its standard form: neither typed, not necessarily
left-linear, confluent, nor constructor-based.
Secondly, we extend the established
resource analysis through a transformational approach, based on earlier work by Avanzini et al.~\cite{ADM15,ADM:2018}.
This provides us with a state-of-the-art tool for the resource analysis of
pure \OCaml~programs, but more generally allows the analysis of general programs.
Technically, the main contributions of the paper are as follows.
\begin{itemize}
\item Employing the standard rewriting semantics in the context of amortised resource
  analysis.
  This standardises the results and simplifies the presentations
  contrasted to related results on amortised analysis of TRSs cf.~\cite{HM:2014,MS:2017b}.
  We emphasise that our analysis does not presuppose termination.
\item We overcome earlier restrictions to typed, completely defined, orthogonal and constructor
  TRSs. Thus, we establish an amortised analysis for standard first-order rewrite system, that is,
  the only restrictions required are the standard restrictions that (i) the left-hand side of a rule
  must not be a variable and (ii) no extra variables are introduced in rules.
\item The analysis is lifted to relative rewriting, that is, the runtime complexity
  of a relative TRS $\relative{\RS}{\RSS}$ is measured by the number of rule
  applications from $\RS$, only.
  This extension is mainly of practical relevance,
  as required to obtain an automation of significant strength.
\item The analysis has been implemented and integrated into \tct. We
  have assessed the viability of the method in context of the \emph{Termination Problem Database} (\emph{TPDB})%
  \footnote{See~\url{http://termination-portal.org/wiki/TPDB}.}
  as well as on an independent benchmark of TRSs.
\item Finally, the aforementioned analysis has been coupled with the complexity reflecting $\hoca$ transformation~\cite{ADM15,ADM:2018} to
  provide a fully automated resource analysis of pure \OCaml. This is based on a refined transformation taking
  types into account to improve the applicabilities of heuristics in the amortised analysis. We provide experimental
  evidence of the viability of the analysis.
\end{itemize}

This paper is structured as follows. In the next section we provide a high-level introduction into the method, exemplifying its use in the analysis of $B^{\star}$-trees. We believe this example to be of independent interest as it constitutes a challenge to other approaches in automated amortised analysis. In Section~\ref{Preliminaries}, we cover
basics. In Section~\ref{Resources}, we introduce the inference system
and prove soundness of the method. In Section~\ref{Implementation}, we detail the implementation of
the method and remark on challenges posed by automation. Section~\ref{Experiments} provides the
experimental assessment of the method. Finally, we conclude in Section~\ref{Conclusion}, where we
also sketch future work.

The paper is a revised and extended version of the original presentation of the established method
in~\cite{MS:2018}. Apart from the usual correction of errors in the conference version, we have incorporated
full formal details and additional examples. Furthermore, the extension to a resource analysis
of higher-order function programs is new, which also triggered a renewed experimental evaluation.

\section{Worst-Case Analysis for Higher Order \OCaml{} Programs}
\label{sec:Motivation}

Term rewrite systems form a generic computational model and are well suited for transformations from
various program languages. Consequently our transformational approach first transforms a given
program into a TRS, whose resource consumption is later analysed (see~\cite{Moser17}).
Consider for instance the following program $\BTreeLookup$, encoding lookup for a variant of
$B^{\star}$-trees in $\OCaml$. $B^{\star}$ trees are data structures usually used for indexing
purposes in databases.
The function lookup traverses through the given $B^{\star}$-tree while searching for an element. The
invariant of the tree is, that each node consists of $n+1$ sub-trees, which are ordered and split
among $n$ values. The data resides in the leaves of the tree. The actual number $n$ specifying the
amount of sub-trees is usually set to a value such that via one disk read a full node including the
separation criteria and links to the sub-trees is retrieved. Thus $n$ is related to the underlying
hardware and therefore kept variable to be able to minimises the number of disk reads.

The worst-case runtime complexity of the program is linear in the size of the input tree, whereas the
best-case complexity is constant due to the exceptions being raised for an invalid tree input.
In case only valid tree inputs are considered the best-case complexity is also linear. Therefore once
dropping lines 11 and 15, which is a reasonable decision under the assumption that an encapsulated
inserting mechanism ensures to generate valid trees only, the worst-case and best-case complexities
coincide.
To the best of our knowledge our tool is currently the only one that that can assess the worst-case and best-case complexity
of the example fully automatically. (See~\cite{MS:2017b} for the amortised analysis for best-case complexity of TRSs.)
%
% \footnote{We have omitted standard exception handlers, assuming only valid input.
% Otherwise, we automatically conclude the correct, but not very informative constant lower bound.}
%
% This is actual OCaml Code!
\begin{lstlisting}
type btree = Leaf of int list | Node of  int list * btree list;;
let rec anyEq nr ys = match ys with
   | [] -> false
   | x::xs -> if nr == x then true else anyEq nr xs;;
exception InvalidTree;;
let rec lookup n node = match node with
  | Leaf(xs) -> anyEq n xs
  | Node(nrs,tss) -> match nrs with
                     | [] -> (match tss with
                              | tGt::[] -> lookup n tGt
                              | _ -> raise InvalidTree)
                     | nr::ns -> match tss with
                                 | t::ts -> if n <= nr then lookup n t
                                            else lookup n (Node(ns,ts))
                                 | [] -> raise InvalidTree;;
\end{lstlisting}

Exploiting the complexity preserving and complexity reflecting $\hoca$ transformation due to Avanzini et
al.~\cite{ADM15,ADM:2018}, we transform the $\OCaml$ program into an equivalent TRS $\RSbtree$.
The transformation employs defunctionalisation of the higher-order program and applies
transformations like inlining, uncurrying and dead code elimination. The output of the
transformation is a TRS with corresponding type information. The latter may be dropped, if we seek a
standard TRS.
Crucially the fully automatic transformation is \emph{complexity reflecting}, that
is, the complexity bounds of the obtained TRS $\RSbtree$ \emph{reflects} the runtime complexity of the above $\OCaml$ program, so that we can employ the downstream resource analysis of the rewrite system $\RSbtree$ to bound the runtime complexity of the functional program.
The resulting TRS $\RSbtree$ (without type information) is given below.

\begin{alignat*}{2}
  && \anyEq(nr,\nil) &\to \false
  \\
  && \anyEq(nr,x \cons xs) &\to \ite({nr=x}, \true, \anyEq(nr,xs))
  \\
  && \lookup(n,\leaf(xs)) &\to \anyEq(n,xs)
  \\
  && \lookup(n,\node(\nil,tGt \cons \color{black}{\nil})) & \to \lookup(n,tGt)
  \\
  && \lookup(n,\node(nr \cons ns,t \cons ts)) & \to \ite(n \leq nr, \lookup(n,t), \lookup(n, \node(ns,ts)))
  \\
  && \lookup(n,\node(\nil,\color{black}{\nil})) & \to \bot
  \\
  && \color{black}{\lookup(n,\node(\nil,t \cons (t1 \cons ts)))} & \to \bot
  \\
  && \lookup(n,\node(t \cons ts,\nil)) & \to \bot
\end{alignat*}

However, the current version of the $\hoca$ prototype relaxes the exception handling and thus
provides the fourth rule as $\lookup(n,\node(\nil,tGt \cons \mathbf{ts})) \to \lookup(n,tGt)$ and
drops the second-to-last rule which becomes unnecessary. Note that this semantic change has no
effect on the worst-case upper bound. Additionally the tool is unable to type the input program
polymorphically. Thus constructors with different type contexts have to be renamed apart. We will be
concerned with these issues in the future, but these are out of scope of this work.

The subsequent analysis of the obtained TRS is based on the amortised analysis
established in the paper. This univariate amortised analysis is based on the potential method
which is coached in syntax-directed annotated type system, where the function arguments
and the function result are annotated with
resource information, see Section~\ref{Resources}.
If the TRS turns out to be \emph{resource bounded}, that is, it is derivable in the inference
system, then the differences in the potential functions before and after the evaluation provide the
worst-case bounds. This is a consequence of the Soundness Theorem, Theorem~\ref{t:1}. Furthermore we
provide suitable constrains on the resource annotations to allow the deduction of polynomial
innermost runtime complexities, cf. Theorem~\ref{t:2}.

The construction is easy to mechanise as the proof search can be encoded into (in general)
non-linear constraints, which can be subsequently handled by an SMT solver embedded in our prototype
implementation. The non-linearity of the generated constraints stems from the fact that the
inductive data types are not restricted and no type information is required a priori, see
Section~\ref{Implementation}. As this is costly, we incorporated further heuristics to obtain linear
constraints, following Hoffmann et al.~\cite{Hoffmann:2011,HAH12b,HDW:2017}. Although these often
speed up the analysis (see Section~\ref{Experiments}), they require comprehensive type information
as they for example utilize the position of the recursive type parameters. $\hoca$ infers following
types for the $\BTreeLookup$ example.

\begin{align*}
  \bool        & \defsym \bot \mid \false \mid \true\\
  \llist(a)    & \defsym \nil \mid a \cons \llist(a)\\
  \nat         & \defsym \zero \mid \mS(\nat)\\
  \tree(a)     & \defsym \lleaf(\llist(a)) \mid \lnode(\llist(a), \treelist(a))\\
  \treelist(a) & \defsym \niltree \mid a \constree \treelist(a)\\
\end{align*}

For instance in the case of the type $\llist(a)$ the heuristics will select a multiple of
the length of the list as appropriate measure for the potential of the given data structure. Thus,
each application of $\cons$ will require a certain potential $x$ and the second parameter
$\llist(a)$ of the list has to pass the potential on, while $\nil$ and the first parameter of
$\cons$ must be typed with potential $0$. This yields $x$ times the length of the list as potential
for this data structure. Clearly in most cases this makes sense. However, in the context of
user-defined data structures this heuristic may lead to infeasibility.

Nonetheless, as in the run of the $\BTreeLookup$ example from above, it is often the case that the
input program is (strongly) typed. To exploit this information, we make use of a recent
implementation of the $\hoca$ transformation that retains the type information. For $\BTreeLookup$
however, the aforementioned heuristics are inapplicable. This is due to the user-defined data
structure which when being processed from top to bottom alternates between nodes of the tree and the
list of possible sub-trees. Hence \raml{}, which only relies on heuristics, reports an unsupported
recursive type error, while
our prototype implementation can derive the linear worst-case complexity of $\BTreeLookup$ in terms
of size of the input tree % in
% the depth of the tree and the maximal number of node successors
fully automatically.

\section{Preliminaries}
\label{Preliminaries}

We assume familiarity with term rewriting~\cite{BN98,TeReSe}
but briefly review basic concepts and notations.

Let $\VS$ denote a countably infinite set of variables and $\FS$ a
signature, such that $\FS$ contains at least one constant.
The set of terms over $\FS$ and $\VS$ is denoted by~$\TERMS$.
We write $\Var(t)$ to denote the set of variables occurring in term $t$.
The \emph{size} $\size{t}$ of a term is defined
as the number of symbols in~$t$.

We suppose $\FS = \CS \dunion \DS$, where
$\CS$ denotes a finite, non-empty set of \emph{constructor symbols},
$\DS$ is a finite set of \emph{defined function symbols}, and $\dunion$ denotes
disjoint union.
Defined function symbols are sometimes referred to as \emph{operators}.
A term $t$ is \emph{linear} if every variable in $t$ occurs only once. A term
$t'$ is the \emph{linearisation} of a non-linear term $t$ if the variables in
$t$ are renamed apart such that $t'$ becomes linear. The notion generalises to
sequences of terms.
A term $t = f(t_1,\dots,t_k)$ is called \emph{basic},
if $f$ is defined, and all $t_i \in \TA(\CS,\VS)$.
We write $\dom(\sigma)$ ($\range(\sigma)$) to denote the domain (range) of $\sigma$.

Let ${\to} \subseteq S \times S$ be a binary relation. We denote by $\to^+$ the
transitive and by $\to^*$ the transitive and reflexive closure of $\to$. By $\to^n$
we denote the $n$-fold application of $\to$. If $t$ is in normal form with respect
to $\to$, we write $s \to^{!} t$. We say that $\to$ is \emph{well-founded} or
\emph{terminating} if there is no infinite sequence $s_0 \to s_1 \to \dots$\ . It
is \emph{finitely branching} if the set $\{ t \mid s \to t\}$ is finite for each
$s \in S$. For two binary relations ${\to_A}$ and ${\to_B}$, the relation of
$\to_A$ \emph{relative} to $\to_B$ is defined by
${\relative{{\to_A}}{{\to_B}}} \defsym {\to_B^* \cdot \to_A \cdot \to_B^*}$.

A \emph{rewrite rule} is a pair $l \to r$ of terms, such that (i) the root symbol
of $l$ is defined, and (ii) $\Var(l) \supseteq \Var(r)$. A \emph{term rewrite
  system} (TRS) over $\FS$ is a finite set of rewrite rules. Observe that
TRSs need not be constructor systems, that is, arguments of left-hand
sides of rules may contain defined symbols. Such function symbols are called
\emph{constructor-like}, as below they will be sometimes subject to similar restrictions as
constructor symbols.

The set of normal forms of a TRS $\RS$ is denoted as $\NF(\RS)$, or $\NF$ for short.
We call a substitution $\sigma$ \emph{normalised with respect to $\RS$} if all terms in the
range of $\sigma$ are ground normal forms of $\RS$. Typically $\RS$ is clear from context, so
we simply speak of a \emph{normalised} substitution.
In the sequel we are concerned with \emph{innermost} rewriting, that is,
an eager evaluation strategy. Furthermore, we consider relative rewriting.

A TRS is \emph{left-linear} if all rules are left-linear, it is \emph{non-overlapping}
if there a no critical pairs, that is, no ambiguity exists in applying rules.
A TRS is \emph{orthogonal} if it is left-linear and non-overlapping.
A TRS is \emph{completely defined} if all ground normal-forms are values.
Note that an orthogonal TRS is confluent. A TRS is \emph{constructor} if
all arguments of left-hand sides are basic.

The \emph{innermost rewrite relation} $\irew$
of a TRS $\RS$ is defined on terms as follows: $s \rsirew t$ if
there exists a rewrite rule $l \to r \in \RS$, a context $C$, and
a substitution $\sigma$ such that $s = C[l\sigma]$, $t = C[r\sigma]$,
and all proper subterms of $l\sigma$ are normal forms of $\RS$.
In order to generalise the innermost rewriting relation
to relative rewriting, we introduce the slightly technical
construction of the \emph{restricted} rewrite relation~\cite{T07}.
The \emph{restricted rewrite relation $\toss{\QS}_{\RS}$}
is the restriction of $\rsrew$ where all arguments of the redex are in
normal form with respect to the TRS $\QS$.
We define the \emph{innermost rewrite relation}, dubbed $\rsirew[\relative{\RS}{\RSS}]$,
of a relative TRS $\relative{\RS}{\RSS}$ as follows.
\begin{equation*}
  {\rsirew[\relative{\RS}{\RSS}]} \defsym
  {{\toss{\RS \cup \RSS}_{\RSS}^{\ast}} \cdot
    {\toss{\RS \cup \RSS}_{\RS}} \cdot
    {\toss{\RS \cup \RSS}_{\RSS}^{\ast}}} \tpkt
\end{equation*}
Observe that ${\rsirew[\RS]} = {\rsirew[\relative{\RS}{\varnothing}]}$ holds.

Let $s$ and $t$ be terms, such that $t$ is in normal-form.
Then a \emph{derivation} $D \colon s \rssrew t$ with respect to a TRS $\RS$ is a finite sequence of rewrite steps.
The \emph{derivation height} of a term $s$ with respect to a
well-founded, finitely branching relation $\to$ is defined as
$\dheight(s,\to) = \max\{ n \mid \exists t \; s \to^n t \}$.

\begin{definition}
  \label{d:runtimecomplexity}
  We define the \emph{innermost runtime complexity}
  (with respect to $\relative{\RS}{\RSS}$):
  $\rc_{\RS}(n) \defsym {} \max \{ \dheight(t,\rsirew[\relative{\RS}{\RSS}]) \mid \text{$t$ is basic and
    $\size{t} \leqslant n$}\}$.
\end{definition}
Intuitively the innermost runtime complexity wrt.\ $\relative{\RS}{\RSS}$ counts
the maximal number of eager evaluation steps in $\RS$ in a derivation over $\RS \cup \RSS$. In the
definition, we tacitly assume that $\rsirew[\relative{\RS}{\RSS}]$ is terminating and finitely
branching.

For the rest of the paper the relative TRS $\relative{\RS}{\RSS}$ and its signature
$\FS$ are fixed. In the sequel of the paper, substitutions are assumed to be
normalised with respect to $\RS \cup \RSS$.

\section{Resource Annotations}
\label{Resources}

In this section, we establish a novel amortised resource analysis for TRSs.
This analysis is based on the potential method and coached in an inference system.
Firstly, we annotate the (untyped) signature by the prospective resource usage
(Definition~\ref{d:annotations}).
Secondly, we define a suitable inference system, akin to a type system.
Based on this inference system we delineate a class of
\emph{resource bounded} TRSs (Definition~\ref{d:resource-bounded}) for which we
deduce polynomial bounds on the innermost runtime complexity for a
suitably chosen class of annotations, cf.~Theorem~\ref{t:2}.

A \emph{resource annotation} $\vec{p}$ is a vector $\vec{p}=(p_1,\dots,p_k)$ over
non-negative rational numbers. The vector $\vec{p}$ is also simply called
\emph{annotation}.
Resource annotations are denoted by $\vec{p}$, $\vec{q}$, $\vec{u}$, $\vec{v}$, \dots, possibly
extended by subscripts and we write $\Vecs$ for the set of such annotations. For resource
annotations $(p)$ of length $1$ we write $p$. A resource annotation does not change its meaning if
zeroes are appended at the end, so, conceptually, we can identify $()$ with $(0)$ and also with $0$.
If $\vec p=(p_1,\dots,p_k)$ we set $k \defsym \len{\vec{p}}$ and
$\max\vec{p} \defsym \max \{p_i \mid i=1,\dots,k\}$. We define the notations
$\vec p\leqslant \vec q$ and $\vec p+\vec q$ and $\lambda \vec p$ for
$\lambda \geqslant 0$ component-wise, filling up with $0$s if needed. So, for
example $(1,2)\leqslant (1,2,3)$ and $(1,2)+(3,4,5)=(4,6,5)$.

\begin{definition}
  \label{d:annotations}
  Let $f$ be a function symbol of arity $n$. We annotate the
  arguments and results of $f$ by \emph{resource annotations}. A (resource) annotation for
  $f$, decorated with $k \in \Qplus$, is denoted as
  $\atypdcl{\vec{p_1} \times \cdots \times \vec{p_n}}{\vec{q}}{k}$.
  The set of annotations is denoted $\TDannot$.
\end{definition}

We lift signatures $\FS$ to \emph{annotated signatures}
$\FS \colon \CS \cup \DS \to (\pow(\TDannot)\setminus \varnothing)$
by mapping a function symbol to a non-empty set of resource annotations.
Hence for any function symbol we allow multiple types. In the context
of operators this is also referred to as \emph{resource polymorphism}.
The inference system, presented below, mimics a type system,
where the provided annotations play the role of types.
If the annotation of a constructor or constructor-like symbol $f$ results in
$\vec{q}$, there must only be exactly one declaration of the form
$\atypdcl{\vec{p_1} \times \cdots \times \vec{p_n}}{\vec{q}}{k}$ in $\FS(f)$, that
is, the annotation has to be \emph{unique}. Moreover, annotations for constructor
and constructor-like symbols $f$ must satisfy the \emph{superposition principle}:
If $f$ admits the annotations
$\atypdcl{\vec{p_1} \times \cdots \times \vec{p_n}}{\vec{q}}{k}$ and
$\atypdcl{\vec{p'_1} \times \cdots \times \vec{p'_n}}{\vec{q'}}{k'}$ then it also
has the annotations
$\atypdcl{{\lambda \vec{p_1}} \times \cdots \times {\lambda \vec{p_n}}}{{\lambda
    \vec{q}}}{\lambda k}$ ($\lambda \in \Qplus$, $\lambda\geqslant 0$) and
$\atypdcl{{\vec{p_1}+\vec{p'_1}} \times \cdots \times
  {\vec{p_n}+\vec{p'_n}}}{{\vec{q}+\vec{q'}}}{k+k'}$.

\begin{example}
  \label{ex:2}
  To exemplify consider the sets $\DS = \{\enq, \rev, \revp, \snoc, \checkF, \head, \tail\}$ and
  $\CS = \{\nil, \cons, \queue, \zero, \mS\}$,
  which together make up the signature $\FS$
  of the motivating example $\RSa$ in Figure~\ref{fig:1}.
  Annotations of the constructors $\nil$ and
  $\cons$ would for example be as follows.
  $\FS(\nil) = \{\atypdcl{}{k}{0} \mid k\geqslant 0\}$ and
  $\FS(\cons) = \{\atypdcl{0 \times k}{k}{k} \mid k\geqslant 0\}$. These annotations
  are unique and fulfill the superposition principle.
\qed
\end{example}

Note that, in view of superposition and uniqueness, the annotations of a given constructor or
constructor-like symbol are  uniquely determined once we fix the resource annotations
for result annotations of the form $(0,\dots,0,1)$
(remember the implicit filling up with $0$s).
An annotated signature $\FS$ is simply called signature, where we sometimes
write $\typed{f}{\atypdcl{\vec{p}_1 \times \cdots \times \vec{p}_n}{\vec{q}}{k}}$
instead of $\atypdcl{\vec{p}_1 \times \cdots \times \vec{p}_n}{\vec{q}}{k} \in \FS(f)$.

The next definition introduces the notion of the potential of a normal form.
For rules $f(l_1,\dots,l_n) \to r$ in non-constructor TRSs the left-hand side
$f(l_1,\dots,l_n)$ need not necessarily be basic terms. However, the arguments
$l_i$ are deconstructed in the rule (\textsf{app}) that we will see in
Figure~\ref{fig:2}. This deconstruction may free potential, which needs to be
well-defined. This makes it necessary to treat defined function symbols in $l_i$
similar to constructors in the inference system (see Definition~\ref{d:freed}).

\begin{definition}
  \label{d:potential}
  Let $v = f(v_1,\dots,v_n)$ be a normal form and let $\qq$ be a resource annotation. We define
  the \emph{potential} of $v$ with respect to $\qq$, written $\Phi(\typed{v}{\qq})$ by cases.
  First suppose $v$ contains only constructors or constructor-like symbols.
  Then the potential is defined recursively.
  \begin{equation*}
    \Phi(\typed{v}{\qq}) \defsym k + \Phi(\typed{v_1}{\pp_1}) + \cdots + \Phi(\typed{v_n}{\pp_n})
    \tkom
  \end{equation*}
  %$\Phi(\typed{v}{\qq}) \defsym k + \Phi(\typed{v_1}{\pp_1}) + \cdots + \Phi(\typed{v_n}{\pp_n})$
  %
  where $\atypdcl{\pp_1 \times \cdots \times \pp_n}{\qq}{k} \in \FS(f)$.
  Otherwise, we set $\Phi(\typed{v}{\qq}) \defsym 0$.
\end{definition}

The \emph{sharing relation}
$\share{\vec{p}}{\vec{p_1},\vec{p_2}}$ holds if $\vec{p_1} + \vec{p_2} = \vec p$.

\begin{lemma}
  \label{l:1}
  Let $v$ be a a normal form.
  If $\share{\vec{p}}{\vec{p_1},\vec{p_2}}$ then
  $\Phi(\typed{v}{\vec{p}}) = \Phi(\typed{v}{\vec{p_1}}) + \Phi(\typed{v}{\vec{p_2}})$.
  Furthermore, if $\vec{p} \leqslant \vec{q}$ then
  $\Phi(\typed{v}{\vec p}) \leqslant \Phi(\typed{v}{\vec q})$.
\end{lemma}
\begin{proof}
  The proof of the first claim is by induction on the structure of $v$. For
  constructor or constructor-like symbols the proof follows the pattern of
  Lemma~10 in~\cite{HM:2014}.
  % GM: space
  We distinguish two cases. Either $v$ is only build from constructor symbols or
  constructor-like symbols. Then by superposition together with uniqueness the
  additivity property propagates to the argument types. For example, if we have the
  annotations $\mS:\atypdcl{6}{12}{3}$, $\mS:\atypdcl{10}{20}{5}$, and
  $\mS:\atypdcl{x}{y}{8}$ then we can conclude $x=16$, $y=22$, for this annotation
  must be present by superposition and there can only be one by uniqueness.
  Otherwise $v$ contains at least one $f \in \DS$,
  but $f$ does not occur as
  argument of a left-hand side in $\RS$. By definition $\Phi(\typed{v}{\qq}) = 0$.
  Thus the lemma holds trivially.
  The second claim follows from the first one and non-negativity of potentials.
\qed
\end{proof}

A \emph{(variable) context} is a partial mapping from variables $\VS$ to
annotations. Contexts are denoted by upper-case Greek letters and depicted as
sequences of pairs $\typed{x}{\qq}$ of variables and annotations, where
$\typed{x}{\qq}$ in a variable context means that the resource $\qq$ can be
distributed over all occurrences of the variable $x$ in the term.

\begin{figure}[t]
  \centering
  \begin{tabular}{c}
    $\infer[(\mathsf{app})]{\tjudge{\typed{x_1}{\pp_1},\dots,\typed{x_n}{\pp_n}}{k}{\typed{f(x_1,\dots,x_n)}{\qq}}}{%
    f \in \CS \cup \DS
    &
      \atypdcl{\pp_1 \times \cdots \times \pp_n}{\qq}{k} \in \FS(f)
      }$
      \qquad \hfill \qquad
      $\infer[(\mathsf{w_{1}})]{\tjudge{\Gamma}{k'}{\typed{t}{\qq}}}{%
      \tjudge{\Gamma}{k}{\typed{t}{\qq}}
    &
      k' \geqslant k
      }$
    \\[2.5ex]
    $\infer[(\mathsf{comp})]{\tjudge{\Gamma_1,\dots,\Gamma_n}{k}{\typed{f(t_1,\dots,t_n)}{\qq}}}{%
    \begin{minipage}[b]{40ex}
      all $x_i$ are fresh\\[.5ex]
      $\tjudge{\typed{x_1}{\pp_1},\dots,\typed{x_n}{\pp_n}}{k_0}{\typed{f(x_1,\dots,x_n)}{\qq}}$
    \end{minipage}
    &
      \begin{minipage}[b]{32ex}
        $k = \sum_{i=0}^n k_i$\\[.5ex]
        $\tjudge{\Gamma_1}{k_1}{\typed{t_1}{\pp_1}} \ \cdots \
        \tjudge{\Gamma_n}{k_n}{\typed{t_n}{\pp_n}}$
      \end{minipage}
    }$%
    \\[2.5ex]
    $\infer[(\mathsf{w_{4}})]{\tjudge{\Gamma, \typed{x}{\pp}}{k}{\typed{t}{\qq}}}{%
    \tjudge{\Gamma}{k}{\typed{t}{\qq}}
    }$
    \qquad \hfill \qquad
    $\infer[(\mathsf{share})]{\tjudge{\Gamma, \typed{z}{\pp}}{k}{\typed{t[z,z]}{\qq}}}{%
    \tjudge{\Gamma, \typed{x}{\rr}, \typed{y}{\vecs}}{k}{\typed{t[x,y]}{\qq}}%
    &
      \share{\pp}{\rr,\vecs}
    &
      \text{$x$, $y$ are fresh}
      }$
    \\[2.5ex]
    $\infer[(\mathsf{w_{2}})]{\tjudge{\Gamma, \typed{x}{\pp}}{k}{\typed{t}{\qq}}}{%
    \tjudge{\Gamma, \typed{x}{\rr}}{k}{\typed{t}{\qq}}
    &
      \pp \geqslant \rr
      }$
      \hfill
      $\infer[(\mathsf{var})]{\tjudge{\typed{x}{\qq}}{0}{\typed{x}{\qq}}}{}$
      \hfill
      $\infer[(\mathsf{w_{3}})]{\tjudge{\Gamma}{k}{\typed{t}{\qq}}}{%
      \tjudge{\Gamma}{k}{\typed{t}{\vecs}}
    &
      \vecs \geqslant \qq
      }$
  \end{tabular}
  \caption{Inference System for Term Rewrite Systems.}
  \label{fig:2}
\end{figure}

\begin{definition}
  \label{d:inferencesystem}
  Our potential based amortised analysis is coached in an inference system whose
  rules are given in Figure~\ref{fig:2}. Let $t$ be a term and $\qq$ a resource
  annotation. The inference system derives judgements of the form
  $\tjudge{\Gamma}{k}{\typed{t}{\qq}}$, where $\Gamma$ is a variable context and
  $k \in \Qplus$ denotes the amortised costs at least required to evaluate $t$.

  Furthermore, we define a subset of the inference rules, free of weakening rules,
  dubbed the \emph{footprint} of the judgement, denoted as
  $\footprint{\Gamma}{k}{\typed{t}{\qq}}$. For the footprint we only
  consider the inference rules (\textsf{app}), (\textsf{comp}), (\textsf{share}), and (\textsf{var}).

  Occasionally we omit the amortised costs from both judgements using the
  notations $\tjudge{\Gamma}{}{\typed{t}{\qq}}$ and
  $\footprint{\Gamma}{}{\typed{t}{\qq}}$.
\end{definition}

To ease the presentation we have omitted certain conditions, like the pairwise
disjointedness of $\Gamma_1,\dots,\Gamma_n$ in the rule (\textsf{comp}), that make
the inference rules deterministic. However, the implementation (see
Section~\ref{Implementation}) is deterministic, which removes redundancy in
constraint building and thus improves performance.
A substitution is called \emph{consistent with $\Gamma$} if
for all $x \in \dom(\sigma)$ if $\tjudge{\Gamma}{}{\typed{x}{\qq}}$,
then $\tjudge{\Gamma}{}{\typed{x\sigma}{\qq}}$. Recall that substitutions
are assumed to be normalised.
Let $\Gamma$ be a context and let $\sigma$ be a substitution
consistent with $\Gamma$. Then $\Phi(\typed{\sigma}{\Gamma}) \defsym \sum_{x \in \dom(\Gamma)}
\Phi(\typed{x\sigma}{\Gamma(x)})$.

\begin{definition}
  \label{d:freed}
  Let $f(l_1,\dots,l_n) \to r$, $n \geqslant 1$, be a rule in the TRS $\relative{\RS}{\RSS}$.
  Further suppose $\typed{f}{\atypdcl{\pp_1 \times \cdots \times \pp_n}{\qq}{k}}$ is
  a resource annotation for $f$ and let $V \defsym \{y_1,\dots,y_m\}$ denote
  the set of variables in the left-hand side of the rule.
  The potential \emph{freed} by the rule is a pair consisting of
  a variable context $\typed{y_1}{\rr_1}, \dots, \typed{y_m}{\rr_m}$ and
  an amortised cost $\ell$, defined as follows:
  \begin{itemize}
  \item The sequence $l'_1,\dots,l'_n$ is a linearisation of $l_1,\dots,l_n$.
    Set $Z \defsym \bigcup_{i=1}^n \Var(l'_i)$ and let $Z = \{z_1,\dots,z_{m'}\}$, where
    $m' \geqslant m$.
  \item There exist annotations $\vecs_1,\dots,\vecs_{m'}$ such that for all $i$ there
    exist costs $\ell_i$ such that
    $\footprint{\typed{z_1}{\vecs_1},\dots,\typed{z_{m'}}{\vecs_{m'}}}{\ell_i}{\typed{l'_i}{\pp_i}}$.
  \item Let $y_j \in V$ and let $\{z_{j_1},\dots,z_{j_o}\} \subseteq Z$ be all
    renamings of $y_j$. Define annotations $\rr_j \defsym \vecs_{j_1} + \cdots + \vecs_{j_o}$.
  \item Finally, $\ell \defsym \sum_{i=1}^n \ell_i$.
  \end{itemize}
\end{definition}

\begin{example}
Consider the rule $\enq(\mS(n)) \to \snoc(\enq(n),n)$ in the
running example, together with the annotated signature $\typed{\enq}{\atypdcl{15}{7}{12}}$.
The left-hand side contains the subterm $\mS(n)$. Using the generic annotation
$\typed{\mS}{\atypdcl{k}{k}{k}}$, the footprint $\footprint{\typed{n}{k}}{k}{\typed{\mS(n)}{k}}$
is derivable for any $k \geqslant 0$. Thus, in particular the rule frees the
context $\typed{n}{15}$ and cost $15$.
\qed
\end{example}

\begin{lemma}
  \label{l:2}
  Let ${f(l_1,\dots,l_n) \to r} \in \relative{\RS}{\RSS}$ and let
  $\typed{c}{\atypdcl{\pp_1 \times \cdots \times \pp_n}{\qq}{0}}$ denote a fresh, cost-free
  constructor. Let $\typed{y_1}{\rr_1}, \dots, \typed{y_m}{\rr_m}$ and $\ell$ be
  freed by the rule. We obtain:
  $\footprint{\typed{y_1}{\rr_1},\dots,\typed{y_m}{\rr_m}}{\ell}{\typed{c(l_1,\dots,l_n)}{\qq}}$.
\end{lemma}
\begin{proof}
  By assumption there exists a linearisation $l'_1,\dots,l'_n$ of the arguments of the left-hand
  side of the rule. By definition no variable occurs twice in the sequence $l'_1,\dots,l'_n$.
  Furthermore, for every $i=1,\dots,n$, the following judgement is derivable:
  \begin{equation}
    \label{eq:1}
    \footprint{\typed{z_1}{\vecs_1},\dots,\typed{z_{m'}}{\vecs_m}}{\ell_i}{\typed{l'_i}{\pp_i}}
    \tkom
  \end{equation}
  where $\ell = \sum_{i=1}^n \ell_i$. Observe that the definition of the
  annotations $\rr_j$ embodies a repeated application of the sharing rule
  \textsf{share}. Thus in order to prove the lemma, it suffices to derive
  $\footprint{\typed{z_1}{\vecs_1},\dots,\typed{z_{m'}}{\vecs_{m'}}}{\ell}{\typed{c(l'_1,\dots,l'_n)}{\qq}}$.
  However, to derive the latter a single composition rule, together with the assumed
  derivations~\eqref{eq:1} suffices.
  \qed
\end{proof}

Based on Definition~\ref{d:freed} we can now succinctly define resource boundedness
of a TRS. The definition constitutes a non-trivial generalisation of Definition~11 in~\cite{HM:2014}.
First the input TRS need no longer be sorted. Second the restriction on constructor TRSs
has been dropped and finally, the definition has been extended to handle relative rewriting.
\begin{definition}
  \label{d:resource-bounded}
  Let $\relative{\RS}{\RSS}$ be a relative TRS, let $\FS$ be a signature and let $f \in \FS$.
  An annotation $\atypdcl{\pp_1 \times \cdots \times \pp_n}{\qq}{k} \in \FS(f)$ is called
  \emph{resource bounded} if for any rule $f(l_1,\dots,l_n) \to r \in \RS \cup \RSS$, we have
  \begin{equation*}
    \tjudge{\typed{y_1}{\rr_1}, \dots, \typed{y_l}{\rr_l}}{k+\ell-\rulecost}{\typed{r}{\qq}}
    \tkom
  \end{equation*}
  % $\tjudge{\typed{y_1}{\rr_1}, \dots, \typed{y_l}{\rr_l}}{k+\ell-\rulecost}{\typed{r}{\qq}}$,
  %
  where $\typed{y_1}{\rr_1}, \dots, \typed{y_l}{\rr_l}$ and $\ell$ are freed
  by the rule if $n \geqslant 1$ and $\ell=0$ otherwise.
  Here, the cost $\rulecost$
  for the application of the rule is defined as follows:
  (i) $\rulecost \defsym 1$ iff ${f(l_1,\dots,l_n) \to r} \in \RS$ and
  (ii) $\rulecost \defsym 0$ iff ${f(l_1,\dots,l_n) \to r} \in \RSS$.
  We call an annotation \emph{cost-free} resource bounded if the cost
  $\rulecost$ is always set to zero.
\end{definition}

  A function symbol $f$ is called \emph{(cost-free) resource bounded} if any
  resource annotation in $\FS(f)$ is (cost-free) resource bounded.
  Finally, $\relative{\RS}{\RSS}$ is called \emph{resource bounded}, or simply \emph{bounded} if any $f \in \FS$ is resource bounded.
Observe that boundedness of $\relative{\RS}{\RSS}$ entails
that the application of rules in the strict part $\RS$ is counted, while the weak part
$\RSS$ is not counted.

In a nutshell, the method works as follows:
Suppose the judgement $\tjudge{\Gamma}{k'}{\typed{t}{\qq}}$ is derivable
and suppose $\sigma$ is consistent with $\Gamma$.
The constant $k'$ is an upper-bound to the amortised cost required for reducing $t$
to normal form. Below we will prove that the derivation height of $t\sigma$ (with
respect to innermost rewriting) is bounded by the difference in the potential
before and after the evaluation plus $k'$. Thus if the sum of the potentials of the
arguments of $t\sigma$ is in $\bO(n^k)$, where $n$ is the size of the arguments and
$k$ the maximal length of the resource annotations needed, then the innermost
runtime complexity of $\relative{\RS}{\RSS}$ lies in $\bO(n^k)$.

More precisely consider the \textsf{comp} rule. First note that this rule is only
applicable if $f(t_1,\dots,t_n)$ is linear, which can always be obtained by the use
of the sharing rule. Now the rule embodies that the amortised costs $k'$ required to
evaluate $t\sigma$ can be split into those costs $k'_i$ ($i\geqslant 1$) required
for the normalisation of the arguments and the cost $k'_0$ of the evaluation of the
operator $f$. Furthermore the potential provided in the context
$\Gamma_1,\dots,\Gamma_n$ is suitably distributed. Finally the potential which
remains after the evaluation of the arguments is made available for the evaluation
of the operator $f$.

Before we proceed with the formal proof of this intuition, we exemplify
the method on the running example.

\begin{example}[continued from Example~\ref{ex:2}]
  \label{ex:3}
  \tct\ derives the following annotations for the operators in the
  running example.
  \begin{alignat*}{4}
    \enq &\colon \atypdcl{15}{7}{12} &\qquad&
    \rev &\colon \atypdcl{1}{0}{4} &\qquad&
    \revp &\colon \atypdcl{1 \times 0}{0}{2}
    \\
    \snoc &\colon \atypdcl{7 \times 0}{7}{14} &\qquad&
    \head &\colon \atypdcl{11}{0}{9} &\qquad&
    \tail &\colon \atypdcl{11}{1}{3}
  \end{alignat*}
\qed
\end{example}

We consider resource boundedness of $\RSa$ with respect to the given (monomorphic)
annotated signatures of Example~\ref{ex:3}.
For simplicity we restrict to boundedness of $\enq$.
We leave it to the reader to check the other cases.
In addition to the annotations
for constructor symbols (cf.~Example~\ref{ex:2})
we can always assume the presence of zero-cost annotations, e.g.\
$\cons \colon \atypdcl{0\times 0}{0}{0}$.
Observe that Rule~6 frees the context $\typed{n}{15}$ and cost $15$.
Thus, we obtain the following derivation.
  \begin{equation*}
  \infer[(\mathsf{share})]{\tjudge{\typed{n}{15}}{26}{\typed{\snoc(\enq(n),n)}{7}}}{%
    \infer[(\mathsf{comp})]{\tjudge{\typed{n_1}{15},\typed{n_2}{0}}{26}{\typed{\snoc(\enq(n_1),n_2)}{7}}}{%
      \infer[(\mathsf{app})]{\tjudge{\typed{q}{7},\typed{m}{0}}{14}{\typed{\snoc(q,m)}{7}}}{%
        \typed{\snoc}{\atypdcl{7 \times 0}{7}{14}}
        }
        &
      \infer[(\mathsf{var})]{\tjudge{\typed{n_2}{0}}{0}{\typed{n_2}{0}}}{}
        &
      \infer[(\mathsf{app})]{\tjudge{\typed{n_1}{15}}{12}{\typed{\enq(n_1)}{7}}}{%
        \typed{\enq}{\atypdcl{15}{7}{12}}
       }
    }
  }
 \end{equation*}
In comparison to~\cite[Example~13]{HM:2014}, where the annotations were found manually,
we note that the use of the interleaving operation~\cite{HM:2014} has been avoided.
This is due to the more general class of annotations considered in our prototype
implementation (see Section~\ref{Implementation}).

The footprint relation forms a restriction of the judgement $\tjudge{}{}{}$ without
the use of weakening. Hence the footprint allows a precise control of the
resources stored in the substitutions, as indicated by the next lemma.

\begin{lemma}
  \label{l:3}
  Let $t$ be a normal form w.r.t.\ $\RS$, where $t$ consists of constructor or
  constructor-like symbols only. If $\footprint{\Gamma}{k}{\typed{t}{\qq}}$, then
  $\Phi(\typed{t\sigma}{\qq}) = \Phi(\typed{\sigma}{\Gamma}) + k$.
\end{lemma}
\begin{proof}
  Let $\Pi$ denote the derivation of the footprint $\footprint{\Gamma}{k}{\typed{t}{\qq}}$
  and let $t = f(t_1,\dots,t_n)$. We proceed by induction on $\Pi$.
  We restrict our attention to the cases where $\Pi$ amounts to a rule application or
  ends in a \textsf{comp} rule. The other cases are treated similarly.

  Suppose $\Pi$ has the following form, so that $t=f(x_1,\dots,x_n)$.
  \begin{equation*}
    \infer{\tjudge{\typed{x_1}{\pp_1},\dots,\typed{x_n}{\pp_n}}{k}{\typed{f(x_1,\dots,x_n)}{\qq}}}{%
      f \in \CS \cup \DS
      &
      \atypdcl{\pp_1 \times \cdots \times \pp_n}{\qq}{k} \in \FS(f)
    }
    \tpkt
  \end{equation*}
  By assumption the annotation $\typed{f}{\atypdcl{\pp_1 \times \cdots \times \pp_n}{\qq}{k}}$
  is unique. Hence, we obtain:
  \begin{equation*}
    \Phi(\typed{f(x_1,\dots,x_n)\sigma}{\qq}) = k + \Phi(\typed{x_1\sigma}{\pp_1}) + \cdots
  + \Phi(\typed{x_n\sigma}{\pp_n}) = \Phi(\typed{\sigma}{\Gamma}) + k \tkom
  \end{equation*}
  from which the claim follows.

  Suppose $\Pi$ ends in a \textsf{comp} rule and thus has the following form.
  \begin{equation*}
    \infer{\tjudge{\Gamma_1,\dots,\Gamma_n}{k}{\typed{f(t_1,\dots,t_n)}{\qq}}}{%
      \tjudge{\overbrace{\typed{x_1}{\pp_1},\dots,\typed{x_n}{\pp_n}}^{=: \Delta}}{k_0}{\typed{f(x_1,\dots,x_n)}{\qq}}
      &
      \tjudge{\Gamma_i}{k_i}{\typed{t_i}{\pp_i}} \quad \text{for all $i=1,\dots,n$}
    }
    \tpkt
  \end{equation*}
  Wlog.\ $t$ is linear. By induction hypothesis we have for all $i=1,\dots,n$:
  $\Phi(\typed{t_i\sigma}{\pp_i}) = \Phi(\typed{\sigma}{\Gamma_i}) + k_i$. We set
  $\rho \defsym \{x_i \mapsto t_i\sigma \mid i=1,\dots,n\}$. Again by induction hypothesis
  we conclude that $\Phi(\typed{f(x_1,\dots,x_n)\rho}{\qq}) = \Phi(\typed{\rho}{\Delta}) + k_0$.
  Now the claim follows as (i) $\Phi(\typed{t\sigma}{\qq}) = \Phi(\typed{f(x_1,\dots,x_n)\rho}{\qq})$,
  (ii) $\Phi(\typed{\rho}{\Delta}) = \sum_{i=1}^n \left( \Phi(\typed{\sigma}{\Gamma_i}) + k_i \right)$
  and (iii) $k = \sum_{i=0}^n k_i$.
  %
  %
  % Suppose $\Pi$ has the following form.
  % %
  % \begin{equation*}
  %   \infer{\tjudge{\Gamma, \typed{z}{\qq}}{k}{\typed{t[z,z]}{\qq}}}{%
  %   \tjudge{\Gamma, \typed{x}{\rr}, \typed{y}{\vecs}}{k}{\typed{t[x,y]}{\qq}}%
  %   &
  %   \share{\pp}{\rr,\vecs}
  %   &
  %   \text{$x$, $y$ are fresh}
  % }
  %   \tpkt
  % \end{equation*}
  % %
  % Then the claim follows by induction hypothesis in conjunction with Lemma~\ref{l:2}.

  % Suppose $\Pi$ has the following form.
  % %
  % \begin{equation*}
  %   \infer{\tjudge{\typed{x}{\qq}}{0}{\typed{x}{\qq}}}{}
  % \end{equation*}
  % %
  % The lemma holds trivially.
\qed
\end{proof}

We state the following substitution lemma. The lemma
follows by simple induction on~$t$.
\begin{lemma}
  \label{l:4}
  Let $\Gamma$ be a context and let $\sigma$ be a substitution consistent
  with $\Gamma$. Then $\tjudge{\Gamma}{}{\typed{t}{\qq}}$ implies
  $\tjudge{}{}{\typed{t\sigma}{\qq}}$.
\end{lemma}

We establish soundness with respect to relative innermost rewriting.

\begin{theorem}
  \label{t:1}
  Let $\relative{\RS}{\RSS}$ be a resource bounded TRS and
  let $\sigma$ be a normalised such that $\sigma$ is consistent
  with the context $\Gamma$.
  Suppose $\tjudge{\Gamma}{k}{\typed{t}{\pp}}$ and
  $t\sigma \rsinrew{K} u\tau$, $K \in \{0,1\}$ for
  a normalising substitution $\tau$. Then there exists a context $\Delta$ such that
  $\tjudge{\Delta}{\ell}{\typed{u}{\qq}}$ is derivable and
  $\Phi(\typed{\sigma}{\Gamma}) + k - \Phi(\typed{\tau}{\Delta}) - \ell \geqslant K$.
\end{theorem}
\begin{proof}
  Let $\Pi$ denote the derivation of the judgement $\tjudge{\Gamma}{k}{\typed{t}{\qq}}$. The
  proof proceeds by case distinction on derivation $D \colon t\sigma \rsinrew{K} u\tau$ and
  side-induction on $\Pi$.

  Suppose $D$ is empty, that is, $t\sigma$ is a normal form wrt.\ $\RS \cup \RSS$.
  We distinguish two subcases. Either (i)
  $t\sigma$ contains only constructor or constructor-like symbols
  or (ii) $t\sigma$ contains at least one defined function symbol which does not occur
  as argument of the left-hand side of a rule in $\RS \cup \RSS$.

  For subcase (i), it suffices to show that
  $\Phi(\typed{\sigma}{\Gamma}) + k \geqslant \Phi(\typed{v}{\qq})$
  holds even under the assumption that no weakening rules are applied in $\Pi$. However,
  due to Lemma~\ref{l:2}, $\footprint{\Gamma}{k}{\typed{t}{\qq}}$ implies that
  $\Phi(\typed{\sigma}{\Gamma}) + k = \Phi(\typed{v}{\qq})$. Thus the theorem follows.
  Now consider sub-case (ii). By definition $\Phi(\typed{v}{\qq}) = 0$
  and the theorem follows as potentials and amortised costs are non-negative.

  Now suppose $D \colon t\sigma \rsirew[\relative{\RS}{\RSS}] u\tau$, that is $D$ is
  non-empty. We exemplify the proof on three subcases.

  For subcase (i), we assume that $\Pi$ has the following form.
  \begin{equation*}
    \infer{\tjudge{\underbrace{\typed{x_1}{\pp_1},\dots,\typed{x_n}{\pp_n}}_{{} = \Gamma}}{k}{\typed{f(x_1,\dots,x_n)}{\qq}}}{%
      f \in \CS \cup \DS
      &
      \atypdcl{\pp_1 \times \cdots \times \pp_n}{\qq}{k} \in \FS(f)
    }
    \tpkt
  \end{equation*}
  Then $\sigma = \{x_i \mapsto v_i \mid i=1,\dots,n\}$, where $v_i \in \NF$. By
  assumption on $D$ there exists a rule $f(l_1,\dots,l_n) \to r$
  and a normalised substitution $\tau$ such that $f(l_1,\dots,l_n)\tau = t\sigma$ and $t\tau = u$. Wlog.\ ${f(l_1,\dots,l_n)} \in \RS$. As $\relative{\RS}{\RSS}$ is bounded
  there exist variables $y_1,\dots,y_m$, resource annotation $\rr_1,\dots,\rr_m$
  and an amortised costs $\ell$ such that the following judgement is derivable.
  \begin{equation*}
    \tjudge{\overbrace{\typed{y_1}{\rr_1}, \dots, \typed{y_l}{\rr_l}}^{{} =: \Delta}}{k+\ell-1}{\typed{r}{\qq}}
    \tpkt
  \end{equation*}
  Due to Lemmata~\ref{l:2} and~\ref{l:3} we obtain
  $\Phi(\typed{\sigma}{\Gamma}) + k = \Phi(\typed{\tau}{\Delta}) + \ell$.
  The theorem follows.

  For subcase (ii), we assume that $\Pi$ has the following form.
  \begin{equation*}
    \infer{\tjudge{\Gamma_1,\dots,\Gamma_n}{k}{\typed{f(t_1,\dots,t_n)}{\pp}}}{%
      \tjudge{\overbrace{\typed{x_1}{\pp_1},\dots,\typed{x_n}{\pp_n}}^{=: \Delta_0}}{k_0}{\typed{f(x_1,\dots,x_n)}{\pp}}
      &
      \tjudge{\Gamma_i}{k_i}{\typed{t_i}{\pp_i}} \quad \text{for all $i=1,\dots,n$}
    }
    \tpkt
  \end{equation*}
  Wlog.\ $t$ is linear. As $t\sigma \rsirew[\relative{\RS}{\RSS}] u\tau$,
  there exist (potentially empty)
  subderivations $D_i \colon t_i\sigma \rsirew[\relative{\RS}{\RSS}] u_i\tau$
  for all $i=1,\dots,n$.
  We set $D_0 \defsym f(x_1,\dots,x_k)\rho \rsirew[\relative{\RS}{\RSS}] u\tau$.

  By side induction hypothesis we conclude the existence of contexts $\Delta_i$
  and annotations $\qq_i$ such that $\tjudge{\Delta_i}{\ell_i}{\typed{u_i}{\qq_i}}$
  and $\Phi(\typed{\sigma}{\Gamma_i}) + k_i \geqslant \Phi(\typed{\tau}{\Delta_i}) + \ell_i$.
  Further let $\rho \defsym \{x_i \mapsto u_i \mid i=1,\dots,n\}$. Then
  again by induction hypothesis, there exists context $\Delta'$ and annotation
  $\qq'$ such that $\tjudge{\Delta'}{\ell_0}{\typed{u}{\qq'}}$ is derivable
  and
  $\Phi(\typed{\rho}{\Delta}) + k_0 \geqslant \Phi(\typed{\tau}{\Delta'}) + \ell_0$.
  Observe that at most one of the inequalities in the potentials is strict. By distinguishing
  all possible subcases, the theorem follows.

  For subcase (iii), we assume $\Pi$ has the following form.
  \begin{equation*}
    \infer{\tjudge{\Gamma, \typed{z}{\qq}}{k}{\typed{t[z,z]}{\qq}}}{%
      \tjudge{\Gamma, \typed{x}{\rr}, \typed{y}{\vecs}}{k}{\typed{t[x,y]}{\qq}}%
      &
      \share{\pp}{\rr,\vecs}
      &
      \text{$x$, $y$ are fresh}
    }
    \tpkt
  \end{equation*}
  Then the theorem follows from the side induction hypothesis in conjunction with Lemma~\ref{l:1}.
\qed
\end{proof}
% \begin{proof}
% Let $\Pi$ denote the derivation of the judgement $\tjudge{\Gamma}{k}{\typed{t}{\qq}}$. The
%   proof proceeds by case distinction on derivation $D \colon t\sigma \rsinrew{K} u\tau$ and
%   side-induction on $\Pi$.
% %
% The proof proceeds by case distinction on $D$ and induction on the
% length of~$\Pi$.
% \qed
% \end{proof}

The next corollary is an immediate consequence of the theorem, highlighting the
connection to similar soundness results in the literature.
\begin{corollary}
  \label{c:1}
  Let $\relative{\RS}{\RSS}$ be a bounded TRS and
  let $\sigma$ be a normalising substitution consistent
  with the context $\Gamma$.
  Suppose $\tjudge{\Gamma}{k}{\typed{t}{\qq}}$ and
  $D \colon t\sigma \rsibrew[\relative{\RS}{\RSS}] v \in \NF$. Then
  (i) $\tjudge{}{}{\typed{v}{\qq}}$ and
  (ii) $\Phi(\typed{\sigma}{\Gamma}) - \Phi(\typed{v}{\qq}) + k  \geqslant \len{D}$
  hold.
\qed
\end{corollary}

The next theorem defines suitable constraints on the resource annotations
to deduce polynomial innermost runtime from Theorem~\ref{t:1}. Its proof
follows the pattern of the proof of Theorem~14 in~\cite{HM:2014}.

\begin{theorem}
  \label{t:2}
  Suppose that for each constructor $c$ with
  $\atypdcl{\pp_1 \times \cdots \times \pp_n}{\qq}{k'} \in \FS(c)$, there exists
  $\rr_i \in \Vecs$ such that $\pp_i \leqslant \qq + \rr_i$ where
  $\max\rr_i \leqslant \max \qq=:r$ and $k' \leqslant r$ with
  $\len{\rr_i} < \len{\qq}=:k$. Then $\Phi(\typed{v}{\qq}) \leqslant r\len{v}^k$,
  and thus the innermost runtime complexity of the TRS under investigation is in
  $\mathcal{O}(n^{k})$.
\end{theorem}
\begin{proof}
  We proceed by induction on $v$. Observe that if $k=0$ then
  $\Phi(\typed{v}{\qq})=0$. Otherwise, we have
  \begin{align*}
    \Phi(\typed{c(v_1,\dots,v_n)}{\qq}) & \leqslant
                                          r + \Phi(\typed{v_1}{\pp_1}) + \cdots + \Phi(\typed{v_n}{\pp_n})
    \\
                                        & \leqslant
                                          r+\Phi(\typed{v_1}{\qq +\rr_1})+\cdots+\Phi(\typed{v_n}{\qq+\rr_n})
    \\
                                        & =
                                          r+\Phi(\typed{v_1}{\qq}) + \Phi(\typed{v_1}{\rr_1})+\cdots+\Phi(\typed{v_n}{\qq})+
                                          \Phi(\typed{v_n}{\rr_n})
    \\
                                        & \leqslant r(1+\len{v_1}^k+\len{v_1}^{k-1}+\dots+\len{v_n}^k+\len{v_n}^{k-1})
                                          \tkom
  \end{align*}
  where we have applied the induction hypothesis in conjunction with
  Lemma~\ref{l:1}.
  The last expressed is bounded by
  $r(1+\len{v_1}+\dots+\len{v_n})^k=r\len{v}^k$ due to the multinomial theorem.
  \qed
\end{proof}
% \begin{proof}
%   The theorem follows the pattern of the proof of Theorem~14 in~\cite{HM:2014}.
% \qed
% \end{proof}

We note that our running example satisfies the premise of Theorem~\ref{t:2}.
Thus the linear bound on the innermost runtime complexity
of the running example $\RSa$ follows.
The next example clarifies that without further assumptions
potentials are not restricted to polynomials.

\begin{example}
  Consider that we annotate the constructors for natural numbers as
  $\typed{\zero}{\atypdcl{}{\vec{p}}{0}}$ and
  $\typed{\mS}{\atypdcl{2\vec{p}}{\vec{p}}{p_1}}$,
  where $\vec{p} = (p_1,\dots,p_k)$.
  We then have, for example, $\Phi(\typed{t}{1})=2^{v}-1$, where $v$ is the value
  represented by $t$.
\qed
\end{example}

\section{Implementation}
\label{Implementation}

In this section we describe the details of important implementation issues. The
realisations of the presented method can be seen twofold. On one hand we have a
standalone program which tries to directly annotate the given TRS.~While on the
other hand the integration into \tct~\cite{avanzini2016tct} uses relative
rewriting. Clearly, as an integration into \tct~was planned from the beginning,
the language used for the implementation of the amortised resource analysis module
is Haskell\footnote{See \url{http://haskell.org/}.}.
% The modular design of \tct~eased the integration tremendously.

The central idea of the implementation is the collection of all signatures and
arising constraints occurring in the inference tree derivations. To guarantee
resource boundedness further constraints are added such that uniqueness and
superposition of constructors (cf.~Section~\ref{Resources}) is demanded and
polynomial bounds on the runtime complexity are guaranteed (cf.~Theorem~\ref{t:2}).

% GM: space
% The rest of this section is structured as follows. First the main idea on how the tool
% applies the inference rules is presented. Then, the implementation of the
% superposition constraints and cost-free inference rule applications are discussed.
% Afterwards, the last part of this section concentrates on heuristics used in
% alternative implementations and their integration in our tool.

\paragraph*{Inference Tree Derivation and Resource Boundedness}

To be able to apply the inference rules the expected root judgement of each rule is generated (as in
Example~\ref{ex:3}) by the program and the inference rules of Figure~\ref{fig:2} are applied. To
gain determinism the inference rules are ordered in the following way. The \textsf{share}-rule has
highest priority, followed by \textsf{app}, \textsf{var}, \textsf{comp} and $\mathsf{w_{4}}$. In
each step the first applicable rule is used while the remaining weakening rules $\mathsf{w_{1}}$,
$\mathsf{w_{2}}$ and $\mathsf{w_{3}}$ are integrated in the aforementioned ones. For each
application of an inference rule the emerging constraints are collected.

To ensure monomorphic typing of function signatures we keep track of a list of signatures. It uses
variables in lieu of actual vectors. For each signature occurrence of defined function symbols the
system refers to the corresponding entry in the list of signatures. Therefore, for each defined
function symbol only one signature is added to the list of signatures. If the function occurs
multiple times, the same references are used. Unlike defined function symbols multiple signature
declarations of constructors are allowed, and thus each occurrence adds one signature to the list.

For the integration into \tct~we utilise the relative rewriting formulation. Instead of requiring
all strict rules to be resource bounded, we weaken this requirement to have at least one strict rule
being actually resource bounded, while the other rules may be annotated cost-free resource bounded.
The SMT solver chooses which rule will be resource bounded. Clearly, this eases the constraint
problem which is given to the SMT solver.

\paragraph*{Superposition of Constructors}

Recall that constructor and constructor-like symbols $f$ must satisfy the superposition principle.
Therefore, for each annotation $\atypdcl{\vec{p_1} \times \cdots \times \vec{p_n}}{\vec{q}}{k}$ of
$f$ it must be ensured that there is no annotation
$\atypdcl{\lambda \cdot \vec{p_1} \times \cdots \times \lambda \cdot \vec{p_n}}{\vec{q'}}{\lambda
  \cdot k}$ with $\lambda \in \Qplus$ and $q \neq \lambda \cdot q'$ in the corresponding set of
annotated signatures. Therefore, for every pair $(q,q')$ with $q' \geqslant q$ and $q > 0$ either
for every $\lambda > 0:~q' \neq \lambda \cdot q$ or if $q' = \lambda \cdot q$ then the annotation
must be of the form
$\atypdcl{\lambda \cdot \vec{p_1} \times \cdots \times \lambda \cdot \vec{p_n}}{\vec{\lambda \cdot
    q}}{\lambda \cdot k}$.

A naive approach is adding corresponding constraints for every pair of return annotations of a
constructor symbol. This leads to universal quantifiers due to the scalar multiplication, which
however, are available as binders in modern SMT solvers~\cite{BarFT-RR-17}. Early experiments
revealed their bad performance. Overcoming this issue using Farkas' Lemma~\cite{farkas1902theorie}
is not possible here. This is due to the fact that a scaling variable for each annotation in every
pair of annotations has to be introduced. Note that the number of constraints increase exponentially
by requiring constraints over pairs of annotations. Thus, we developed a heuristic of spanning up a
vector space using unit vectors for the annotation of the return types for each constructor. Each
annotated signature of such a symbol must be a linear combination of the base signatures.

Both methods, universal quantifiers and base signatures lead to non-linear constraint problems.
However, these can be handled by some SMT solvers\footnote{We use the SMT Solvers z3
  (\url{https://github.com/Z3Prover/z3/wiki}) and MiniSmt
  (\url{http://cl-informatik.uibk.ac.at/software/minismt/}).}. Thus, in contrast to the techniques
presented in~\cite{Hoffmann:2011,HAH12b,HDW:2017}, which restrict the potential function to
predetermined data structures, like lists or binary trees, our method allows any kind of data
structure to be annotated.

\begin{example}
  \label{ex:basectrs}
  For instance consider the base constructors
  $\cons_{1} \colon \atypdcl{(0,0) \times (1,0)}{(1,0)}{1}$ and
  $\cons_{2} \colon \atypdcl{(0,0) \times (2,1)}{(0,1)}{1}$ for a constructor
  $\cons$. An actual instance of an annotated signature is
  $n_{1} \cdot \cons_{1} +~n_{2} \cdot \cons_{2}$ with $n_{1},n_{2} \in \N$. As the
  return types can be seen as unit vectors of a Cartesian coordinate system the
  superposition and uniqueness properties hold.
\qed
\end{example}

\paragraph*{Cost-Free Function Symbols}

Inspired by Hoffmann~\cite[p.93ff]{Hoffmann:2011} we additionally implemented a
cost-free inference tree derivation when searching for non-linear bounds. The idea
is that for many non-tail recursive functions the freed potential must be the one
of the original function call plus the potential that gets passed on.

The inference rules are extended by an additional \textsf{app}-rule, which
separates the function signature into two parts, cf. Figure~\ref{fig:costfree}.
On the left there are the monomorphic and cost-free signatures while on the right a
cost-free part is added. For every application of the rule the newly generated
cost-free signature annotation must be cost-free resource bounded, for this the
cost-free type judgement indicated has to be derived for any rule
$f(l_1,\dots,l_n)~\to~r$ and freed context
$\typed{y_1}{\rr_1}, \dots, \typed{y_l}{\rr_l}$ and cost $\ell$.
Thus, the new set of annotations for a defined function symbols $f$ is given by
the following set, cf.~\cite[p.~93]{Hoffmann:2011}.
\begin{equation*}
  \{ \atypdcl{\pp_{1} + \lambda \cdot \ppcf_{1} \times \cdots \times
  \pp_{n} + \lambda \cdot \ppcf_{n}}{\qq + \lambda \cdot \qqcf}{k + \lambda \cdot
  \kcf}\ |\ \lambda \in \Qplus, \lambda \geqslant 0 \}
\tpkt
\end{equation*}
%
% Similar to~\cite{Hoffmann:2011} the inference system is the same as the one given
% in Figure~\ref{fig:2}, with the exception that a second function application rule
% \textsf{app} is added. It is given in Figure~\ref{fig:costfree} and states that the
% function signature annotation is separated in two parts. Firstly, the base
% annotation which is the monomorphically typed signature annotation as before, and
% the second part is a cost-free resource bounded annotation. Therefore, when using
% the cost-free function application inference rule the set of annotated signatures
% for this application of the function $f$ is of the form
% $\FS(f) = \{ \atypdcl{\pp_{1} + \lambda \cdot \ppcf_{1} \times \cdots \times
%   \pp_{n} + \lambda \cdot \ppcf_{n}}{\qq + \lambda \cdot \qqcf}{k + \lambda \cdot
%   \kcf}\ |\ \lambda \in \Qplus, \lambda \geqslant 0 \}$~\cite[p.93]{Hoffmann:2011}.
%
The decision of which \textsf{app} rule is
applied utilises the strongly connected component (SCC) of the call graph analysis
as done in~\cite[p.93ff]{Hoffmann:2011}.

\begin{figure}[t]
  \centering
  \begin{tabular}{c}
    $\infer{\tjudge{\typed{x_1}{\pp_1},\dots,\typed{x_n}{\pp_n}}{k}{\typed{f(x_1,\dots,x_n)}{\qq}}}{%
      \begin{minipage}[b]{40ex}
      $\atypdcl{\ppcf_{1} \times \cdots \times \ppcf_{n}}{\qqcf}{\kcf} \in \FSCF(f)\\$
      $\atypdcl{\pp_1 \times \cdots \times \pp_n}{\qq}{k} \in \FS(f)$
      \end{minipage}
    &
      \begin{minipage}[b]{40ex}
        \centering
        $\tjudge{\typed{y_1}{\rr_1}, \dots,
          \typed{y_l}{\rr_l}}{k+\ell}{\typed{r}{\qq}}$\\[.5ex]
      \end{minipage}
    }$
  \end{tabular}
  \caption{Additional \textsf{app} Rule for Cost-Free Derivation, where $f \in \CS \cup \DS$. }
  \label{fig:costfree}
\end{figure}

Suppose $f$ is the function being
analysed. Whenever, a currently deriving function $g$, where $g$ is not
constructor-like, resides in the SCC of $f$ and the current derivation is not the
cost-free analysis of $f$, the cost-free application rule of
Figure~\ref{fig:costfree} is used. In all other cases, the original function
application rule of Figure~\ref{fig:2} is used.

We have experimented with this heuristic, by pruning the need of the SCC
requirement, such that the cost-free inference rule for application is not just
allowed when $g$ resides in the SCC of $f$, but also when $g$ is reachable
from $f$ in the call graph. This adaption however, increased the constraint problem
size tremendously on one hand which obviously resulted in longer execution times
but on the other hand could only infer one new cubic example and move one example
from a cubic polynomial upper bound to quadratic. This makes sense, as the
motivation behind this extension lies in the restricted feasibility of
monomorphically annotations of recursive calls.

The implementation of the cost-free derivations uses a new set of signatures for which defined
function symbol and constructor symbol signatures can occur multiple times. However, as above for
each inference tree derivation only a single instance of a signature for each defined function
symbol may be used. % Nonetheless, for each cost-free derivation different annotations of cost-free
% signatures of defined function symbols may be used.

\paragraph*{Alternative Implementation of the Superposition Principle (Heuristics)}

Similar to~\cite{Hoffmann:2011,HM:2014} we integrated the additive shift $\shift(\vec{p})$
and interleaving $\vec p\interleave \vec q$ for constructors when type information is given.
Here $\shift (p_1,\dots,p_k) \defsym (p_1+p_2,p_2+p_3,\dots, p_{k-1} + p_k, p_k)$
and $\vec p\interleave \vec q \defsym (p_1,q_1,p_2,q_2,$ $\dots,p_k,q_k)$, where
the shorter of the two vectors is padded with $0$s.
These heuristics are designed such that the superposition principle holds,
without the need of base annotations. Therefore, the constraint problem automatically becomes linear
whenever these heuristics are used which tremendously reduces the execution
times.

% However, according to the experiments (see the detail results online)
% these heuristics are only rarely applicable
% and often require comprehensive type information. This additional information
% allows to separate constructors named alike but with different types. For instance,
% a list of lists can then have different base annotations compared to a simple list,
% even though the constructors have the same name.
% %
% The rather poor performance of these heuristics in the presence of only generic type
% information came as a surprise to us. However, in hindsight it clearly showcases
% the importance of comprehensive type information (as e.g.\ demanded by~\raml) for
% the efficiency of automation of resource analysis in functional programming.

\section{Experimental Evaluation}
\label{Experiments}

\begin{figure}[t]
  \footnotesize
  \centering
\begin{tabular}{p{0.29\textwidth} | x{0.08\textwidth} x{0.08\textwidth} x{0.08\textwidth}
    x{0.08\textwidth} x{0.08\textwidth} x{0.08\textwidth} }
     Testbed: TPDB & \(\bO(1)\) & \(\bO(n^{1})\) & \(\bO(n^{2})\) & \(\bO(n^{3})\) & \(\bO(n^{\geqslant4})\) & Fail\\
    \hline
    \tct{} with ARA & 44 & 211 & 71 & 27 & 10 & 648\\
    \tct{} without ARA & 44 & 210 & 68 & 28 & 8 & 653 \\
    Standalone ARA & 14 & 105 & 30 & 0 & 0 & 862 \\
    Standalone ARA with CF & 14 & 105 & 46 & 8 & 0 & 838\\
    \aprove{} & 53 & 321 & 127 & 35 & 12 & 463\\
    \hline
    \tct{} with ARA & 3.55 & 9.87 & 17.32 & 26.02 & 43.39 & 50.01 \\
    \tct{} without ARA & 3.13 & 9.90 & 14.48 & 24.85 & 40.46 & 49.79 \\
    Standalone ARA & 0.10 & 0.24 & 0.81 & 0.00 & 0.00 & 2.44 \\
    Standalone ARA with CF & 0.15 & 1.00 & 8.65 & 10.94 & 0 & 18.54\\
    \aprove{} & 2.03 & 34.43 & 151.77 & 190.17 & 199.65 & 252.52 \\
  \hline
  \end{tabular}
  \caption{Experimental evaluation of \tct{}~with ARA, \tct{}~without ARA and standalone ARA with
    and without cost-free (CF) derivations on the TPDB for number of instances (top) and required
    time in seconds (bottom).}
  \label{fig:res-tpdb}
\end{figure}

\begin{figure}[t]
  \footnotesize
  \centering
  \begin{tabular}{p{0.29\textwidth} | x{0.08\textwidth} x{0.08\textwidth} x{0.08\textwidth}
    x{0.08\textwidth} x{0.08\textwidth} x{0.08\textwidth} }
     Testbed: TRS programs  & \(\bO(1)\) & \(\bO(n^{1})\) & \(\bO(n^{2})\) & \(\bO(n^{3})\) & \(\bO(n^{\geqslant4})\) & Fail\\
    \hline
    \tct{} with ARA & 2 & 45 & 24 & 8 & 2 & 35\\
    \tct{} without ARA & 2 & 42 & 26 & 7 & 2 & 37 \\
    Standalone ARA & 1 & 43 & 2 & 0 & 0 & 70 \\
    Standalone ARA with CF & 1 & 43 & 14 & 1 & 0 & 57 \\
    \hline
    \tct{} with ARA & 0.05 & 1.92 & 10.10 & 27.20 & 51.65 & 15.11 \\
    \tct{} without ARA & 0.02 & 1.75 & 8.81 & 21.83 & 42.76 & 15.58 \\
    Standalone ARA & 0.04 & 0.23 & 0.14 & 0.00 & 0.00 & 0.25 \\
    Standalone ARA with CF & 0.08 & 1.56 & 8.97 & 1.21 & 0.00 & 54.91 \\
    \hline
  \end{tabular}
  \caption{Experimental evaluation of \tct{}~with ARA, \tct{}~without ARA and
    standalone ARA with and without cost-free (CF) derivations on the TRS testbed for number of
    instances (top) and required time in seconds (bottom).}
  \label{fig:res-testbed}
\end{figure}

% \MS{check if we explained all seen TRS}

In this section we evaluate the prototype implementation of the amortised analysis and how it deals
with some selected examples including the paper's running example \textsf{queue}. All
experiments\footnote{Detailed data is available at
  \url{http://cl-informatik.uibk.ac.at/software/tct/experiments/ara\_worstcase/.}} were conducted on a
machine with an Intel(R) Core(TM) i7-7700 CPU @ 3.60GHz and 16GB RAM. The timeout was set to 60
seconds.
The evaluation is split in the core application of the prototype on first-order TRSs as well as the
assessment of the combination of the provided prototype in connection with the \hoca{}
transformation. The latter is able to verify the worst-case bounds of higher-order functional
programs fully automatically.

Thus for benchmarking we first elaborated the amortised resource analysis (ARA) directly on the
TRSs. As these tests directly evaluate the applicability and functionality of ARA for TRSs we
provide results for two separate testbeds. First we utilize the
runtime-complexity-innermost-rewriting folder of the TPDB\footnote{We refer to Version~10.4 of the
  Termination Problem Database, available
  from~\url{http://cl2-informatik.uibk.ac.at/mercurial.cgi/TPDB}.} which includes 1011 problems, and
second a collection consisting of 116 TRSs representing first-order functional
programs~\cite{Glenstrup:1999,Frederiksen:2002}, transformations from higher-order
programs~\cite{ADM15,ADM:2018}, or~\raml\ programs~\cite{HAH12b} and interesting examples from the
TPDB. We will refer to the later as TRS testbed, while the former is called TBDB testbed for which
we use as a comparison the results of \aprove{}\footnote{See
  \url{https://aprove-developers.github.io/trs_complexity_via_its/} for detailed results of \aprove.
  Timeout: 300 seconds, Intel Xeon with 4 cores at 2.33 GHz and 16GB of RAM.}. For both testbeds we
compared the competition version of \tct{} without the amortised resource analysis (ARA) to a
version where we integrated ARA. The integration into \tct{} is configured to not use any cost-free
derivations (CF) as the toolchain of \tct{} already simplifies the problems such that the additional
constraints become an almost always useless overhead. Thus to provide insights in the cost-free
derivation mechanism we compare the results of standalone ARA with and without the cost-free
derivations as well.

Figure~\ref{fig:res-tpdb} presents the experimental evaluation on the TPDB testbed. As expected ARA
can solve less examples than \tct{} which makes sense as it includes several techniques like
polynomial interpretations, matrix interpretations~\cite{MSW08} or dependency
pairs~\cite{HM:2008,arts2000termination}. Nonetheless ARA, which is able to infer univariate bounds
only, can solve a reasonable amount of examples. It can be seen that ARA with cost-free derivations
can solve more examples in the quadratic and cubic case to the expense of increased running time.
Recall that cost-free derivations are important for non-tail recursive functions. This extension to
ARA increases the applicability of the analysis tremendously. For instance in the quadratic case the
prototype is able to increase the solved instances by more than 50\%. However, also note that the
increase of running time for ARA with cost-free derivations as compared to ARA without cost-free
derivations increases non-linearly. This is due to the additional constraints introduced. As can be
seen in the table \aprove{} is able to solve most instances, where the time required to solve the
problems increases with increasing complexity classes.

When comparing \tct{} with ARA and \tct{} without ARA one can see that $5$ more examples can be
solved. This includes examples from the $\raml{}$ folder, like \textsf{quicksort},
\textsf{bitvectors} and \textsf{depth/breadth first search}. Furthermore, sometimes better bounds
can be inferred with integrated ARA. For instance \tct{} without ARA infers a quadratic bound for
the example \textsf{\#3.42} which computes the binary representation of numbers, while with ARA the
correct linear bound can be inferred.
However, by adding this additional method to the toolchain of \tct{} we observed a slightly
increased average running time. Nonetheless the actual running time is closely problem related. To
exemplify, for the above mention program \textsf{\#3.42} the running time decreases by almost a
factor of $10$ when ARA is added (from $21.45s$ to $2.61s$), while for the example
\textsf{appendAll} the running time increases by a factor of $2.5$ (from $1.05s$ to $2.72s$).

Figure~\ref{fig:res-testbed} provides the results of the experimental evaluation on the TRS testbed.
Again the number of solved instances increases by allowing cost-free derivations as can be seen on
the standalone ARA results. However, the running times increase as well. Furthermore, it can be seen
that standalone ARA is able to infer more examples with linear worst-case bounds as \tct{} without
ARA. However, for higher polynomial bounds the toolchain of \tct{} is required. As the impact of ARA
is higher on this testbed we can see that the methods works best in the setting of program
translations. \tct{} with ARA can solve the examples \textsf{bfs.raml} and \textsf{decrease} as
opposed to \tct{} without ARA. Furthermore, better bounds for \textsf{\#3.42} and \textsf{queue} can
be found. In terms of running times we again observed a slight increase when ARA was added to
\tct{}.

For the second type of evaluation we collected combination of selected $\OCaml$ programs from both
the $\hoca$~\cite{ADM15} and the \raml{} testbed, complemented with interesting examples. Overall
the testbed includes $48$ programs. For this evaluation category we analyse the capabilities of ARA
in connection with the \hoca{} transformation and compare the results to \raml{}, which uses an
amortised analysis based on the $\OCaml$ code directly. The results of the evaluation are summarised
in Figure~\ref{fig:res-hoca}. To provide further insights we tested the method in various
combinations. First we used the tool \tct{} $\hoca{}$, which translates the input $\OCaml$ program
to a TRS and then subsequently calls \tct{}. However, the type information from the $\OCaml$ code is
lost during the translation process as the \tct{} analysis in the current form is type-less. Again
we evaluate the tool with and without the integration of ARA. Nonetheless, as the type information
can be utilized by the heuristics we integrated the transformation process into ARA as well and
compare the results of standalone ARA on the (pre-translated) TRS without type information (TRS) to
the results of standalone ARA with the integration of the translation and thus with type information
(HOCA). Finally we compare the results to the latest version of \raml{} (of July 2018). Recall that
\raml{} is able to infer multivariate bounds, while ARA is built for univariate bounds only.

When evaluating the standalone ARA version we observed no difference in terms of solved instances.
The only difference provided in the table is one additional solved instance for standalone ARA on
the TRS. This is due to the reason that the integrated transformation takes more than 60 seconds
(about $300s$) and thus the tool runs into a timeout. Without timeout the same number of solved
instances are gained. This makes sense, as the heuristics only reduce the constraints to a linear
problem (LP), while not providing any strength. However, in most cases the running time decreases
when type information is available as the constraint problem can be solved without backtracking
mechanism. This can be seen best in the column for quadratically inferred worst-case bounds, where
the running time decreases by almost $10s$. Furthermore, as ARA is based on base vectors it always
tries to improve the quality of the solution by using non-linear constraint solving while \raml{}
relies on heuristics only. This explains the huge gap of running time between \raml{} and standalone
ARA. Finally when looking at the \tct{} $\hoca{}$ results we observe that the version with ARA can
find a better bound for the example \textsf{calculator.raml}. Furthermore, for this testbed \tct{}
$\hoca{}$ provides comparable results to \raml{}. However, the solved instances are different to a
great extend. For instance, while all tools but \raml{} are able to solve examples like
\textsf{avanzini\_original.hoca}, \textsf{btree\_lookup.hoca}, or \textsf{flip.hoca}, only \raml{}
is able to provide bounds for $13$ examples including \textsf{quicksort.raml},
\textsf{bigints\_add.raml}, or \textsf{isort.raml} (see detailed results). These results are due to
several reasons. First \raml{} does not support the user-defined type in \textsf{btree\_lookup.hoca}
as discussed above. Then \raml{} does not support lazy evaluation, while the $\hoca{}$
transformation does. Third, for large programs the transformation process of \tct{} $\hoca{}$ can
take long which results in the fact that either no or only little time is left for the analysis.
In the following we will investigate some of the programs % of the testbeds
in detail.

\begin{figure}[t]
  \footnotesize
  \centering
  \begin{tabular}{p{0.29\textwidth} | x{0.08\textwidth} x{0.08\textwidth} x{0.08\textwidth}
    x{0.08\textwidth} x{0.08\textwidth} x{0.08\textwidth} }
    Testbed: $\hoca{}$  & \(\bO(1)\) & \(\bO(n^{1})\) & \(\bO(n^{2})\) & \(\bO(n^{3})\) & \(\bO(n^{\geqslant4})\) & Fail\\
    \hline
\tct{} $\hoca{}$ with ARA    & 1 & 17 & 8 & 1 & 2 & 19\\
\tct{} $\hoca{}$ without ARA & 1 & 17 & 7 & 2 & 2 & 19\\
Standalone ARA CF (TRS)      & 1 & 18 & 5 & 0 & 0 & 24\\
Standalone ARA CF (HOCA)     & 1 & 17 & 5 & 0 & 0 & 25\\
Raml-1.4.1                   & 1 & 12 & 7 & 8 & 1 & 19\\
    \hline
\tct{} $\hoca{}$ with ARA    & 0.09 & 1.53 & 7.00 & 8.04 & 35.81 & 57.28\\
\tct{} $\hoca{}$ without ARA & 0.03 & 0.59 & 4.70 & 30.35 & 5.34 & 57.21\\
Standalone ARA CF (TRS)      & 0.10 & 0.38 & 49.02 & 0.00 & 0.00 & 46.02\\
Standalone ARA CF (HOCA)     & 0.07 & 0.78 & 39.50 & 0.00 & 0.00 & 46.66\\
Raml-1.4.1                   & 0.03 & 0.07 & 0.79 & 1.10 & 3.38 & 12.12\\
    \hline
  \end{tabular}
  \caption{Experimental evaluation of \tct{} $\hoca{}$~with ARA, \tct{} $\hoca{}$~without ARA and
    standalone ARA with cost-free (CF) derivations on the translated TRS (without type information)
    and directly on the $\OCaml$ code which calls the $\hoca{}$ transformation. The tools were
    evaluated on the $\hoca{}$ testbed for instances (top) and required time in seconds (bottom).}
  \label{fig:res-hoca}
\end{figure}

%
% ,
% as well as the output of \aprove{} as presented in~\cite{naaf2017complexity}\footnote{See
%   \url{https://aprove-developers.github.io/trs_complexity_via_its/} for detailed results of \aprove.
%   Timeout: 300 seconds, Intel Xeon with 4 cores at 2.33 GHz and 16GB of RAM.}.
%
% Figure~\ref{fig:results2} shows the results of the experiments conducted for the \tct{} with ARA. In
% a companion paper, we have studied \emph{best case} complexity for which we suitably adapted the
% amortised resource analysis to obtain lower bounds on the best case complexity. Therefore, the
% standalone tool is also able to infer best case complexity bounds for TRSs~\cite{MS:2017b}.

\paragraph*{\#3.42 -- Binary representation}

Given a number $n$ in unary encoding as input, the TRS computes the binary
representation $(n)_2$ by repeatedly halving $n$ and computing the last bit, see
the Appendix for the TRS.
The optimal runtime complexity of $\RSa$ is linear in $n$. For this, first observe
that the evaluation of $\half(\mS^m(\zero))$ and $\lastbit(\mS^m(\zero))$ requires
about $m$ steps in total. Secondly, $n$ is halved in each iteration and thus the
number of steps can be estimated by $\sum_{i=0}^k 2^i$, where
$k \defsym \len{(n)_2}$. As the geometric sum computes to $2 \cdot 2^k - 1$, the
claim follows.
Such a precise analysis is enabled by an amortised analysis, which takes the
sequence of subsequent function calls and their respective arguments into account.
Compared to former versions of \tct{} which reported $\bO(n^2)$ we find this
optimal linear bound of $\bO(n)$ when ARA is enabled. Furthermore, the best case
analysis of ARA shows that this bound is tight by returning $\Omega(n)$. Similarly
\aprove~\cite{Giesl:2017} yields the tight bound employing a size abstraction
to \emph{integer transition systems} (\emph{ITSs} for short), cf.~\cite{naaf2017complexity}.
The resulting ITSs are then solved with~\cofloco~\cite{florescofloco},
which also embodies an amortisation analysis.

\paragraph*{insertionsort.raml}
Insertionsort from the TPDB has quadratic runtime complexity. ARA with cost-free derivations enabled
is able to infer this bound. Similarly can \tct{} with ARA using the default setup, while when
looking for the best bound only it is unable to handle the trade off between execution time and
tightness of the bound and runs into a timeout. \aprove{} infers $\bO(n^{2})$ as well. This bound is
tight~\cite[p.158ff]{Hoffmann:2011}. The best case analysis finds a linear lower bound for this
implementation of insertionsort.
%
% not correct anymore
% --------------------
% \textsf{splitandsort.raml} first groups the input by a specified key and then sorts
% each grouped list using quicksort. The optimal runtime complexity for this program
% is $\bO(n^{2})$~\cite[158ff]{Hoffmann:2011}. Although far from being optimal, \tct{}
% with ARA is able to find the worst-case upper bound $\bO(n^{5})$, whereas \aprove{}
% infers a cubic bound.

% \paragraph*{mergesort.raml}

\paragraph*{tpa2 -- Multiple Subtraction} %
This TRS from the TPDB iterates subtraction until no more rules can be applied. The latest version
of \tct{} with ARA is in contrast to \tct{} without ARA able to solve the problem. The inferred
quadratic worst-case bound coincides with the bounds provided by \aprove{}.

\paragraph*{bfs.raml -- Depth/Breadth-First Search} This program, found in the TPDB, is a
translation of depth-first search (DFS) and breadth-first search (BFS) from \raml~syntax, see
Figure~\ref{lst:bfs} for the \raml~code. Note that the TRS uses strict rules for the equality check
which recurses on the given data structure. In DFS a binary tree is searched one branch after the
other for a matching entry while BFS uses two lists to keep track of nodes of a binary tree to be
visited. The first one is used to traverse on the nodes of the current depth, whereas the second
list collects all nodes of the next depth to visit. After each iteration the futurelist is reversed.
Further, note that BFS is called twice in the function \textsf{bfs2}. Standalone ARA and \tct{} with
ARA are the only tools which are able to infer a complexity bound of $\bO(n^{2})$.

\begin{figure}[tb!]
  \begin{lstlisting}[language=Caml]
  dfs (queue,x) = match queue with
     | [] -> leaf
     | (t::ts) -> match t with
                   | leaf -> dfs(ts,x)
                   | node(a,t1,t2) ->
                      if a == x then t else dfs(t1::t2::ts,x);
  dodfs (t,x) = dfs([t],x);
  bfs(queue,futurequeue,x) = match queue with
     | [] -> match futurequeue with
        | [] -> leaf
        | (t::ts) -> bfs(reverse(t::ts),[],x)
     | (t::ts) -> match t with
        | leaf -> bfs(ts,futurequeue,x)
        | node(y,t1,t2) -> if x==y then node(y,t1,t2)
                           else bfs(ts,t2::t1::futurequeue,x);
  dobfs(t,x) = bfs([t],[],x);
  bfs2(t,x) = let t' = dobfs(t,x) in dobfs(t',x);
\end{lstlisting}
    \caption{DFS and BFS in RaML Syntax~\cite[p.70]{Hoffmann:2011} of which the
      translation to an TRS can be found in the TPDB.}
    \label{lst:bfs}
\end{figure}

% \paragraph*{Experimental Evaluation}
% We have conducted several further experiments on the TPDB, as well as on the
% smaller testbed composed of interesting examples with the focus on program
% translations. Over the last year the strategy of \tct{} was adapted to focus on
% TRSs which were translated from functional programs. Thus, the examples which can
% be solved are distinct from the \tct{} competition strategy of 2016 to a great
% extent. Due to ARA the latest competition strategy of \tct{} can solve 5 more
% examples of the TPDB than without ARA and for 14 examples a better bound can be
% inferred. On the small testbed \tct{} with ARA can find better bounds for 22
% examples in contrast to \tct{} without ARA and additionally \textsf{bfs.raml} can
% be solved. For further experiments see the detailed results.% :

\section{Conclusion}
\label{Conclusion}

In this paper we have established a novel automated amortised cost analysis for
term rewriting. In doing so we have not only implemented the methods detailed in
earlier work~\cite{HM:2014}, but also generalised the theoretical basis
considerably. We have provided a prototype implementation and integrated into~\tct.

More precisely, we have extended the method of amortised resource analysis to \emph{unrestricted}
term rewrite systems, thus overcoming typical restrictions of functional programs
like left-linearity, pattern based, non-ambiguity, etc.
This extension is non-trivial and generalises earlier results in the literature.
Furthermore, we have lifted the method to relative rewriting.
The latter is the prerequisite to a \emph{modular} resource analysis,
which we have provided through the integration into~\tct.
The provided integration of amortised resource analysis into~\tct\ has led to
an increase in overall strength of the tool (in comparison to the latest version
without ARA and the current version of~\aprove). Furthermore in a significant
amount of cases we could find better bounds than before.

In future work we want to focus on lifting the provided amortised analysis in two
ways. First we want to extend the provided univariate analysis to a multivariate
analysis akin the analysis provided in~\raml. The theoretical foundation for this
has already been provided by Hofmann et al.~\cite{HM:2015}. However efficient
automation of the method proposed in~\cite{HM:2015} requires some sophistication.
%       %
Secondly, we aim to overcome the restriction to constant amortised analysis and
provide an automated (or at least automatable) method establishing logarithmic
amortised analysis. This aims at closing the significant gap of existing methods in
contrast to the origin of amortised analysis~\cite{ST:1983,Tarjan:1985}, compare
also~\cite{Nipkow:2015}.

\bibliographystyle{elsarticle-num}
\bibliography{references}

\end{document}